\documentclass[conference]{IEEEtran}

\usepackage{amsmath,amssymb,amsthm}
\usepackage{mathtools}

\usepackage{graphicx}
\usepackage{epsfig} 

\usepackage{array}
\usepackage{booktabs}

\usepackage{cite}

\usepackage{algorithmic}
\usepackage[linesnumbered,ruled,vlined]{algorithm2e}

\usepackage{xcolor}
\newtheorem{theorem}{Theorem}

\newtheorem{proposition}{Proposition}

\theoremstyle{definition}
\newtheorem{definition}{Definition}
\theoremstyle{remark}

\usepackage{comment}

\newcommand{\defeq}{\vcentcolon=}
\newcommand{\LT}{L^{\raisebox{0.2ex}{\ensuremath{\scriptscriptstyle (T)}}}}

\newcommand{\RT}{R^{\raisebox{0.2ex}{\ensuremath{\scriptscriptstyle T}}}}
\newcommand{\YT}{Y^{\raisebox{0.2ex}{\ensuremath{\scriptscriptstyle T}}}}
\newcommand{\YTone}{Y^{\raisebox{0.2ex}{\ensuremath{\scriptscriptstyle T-1}}}}

\newcommand{\YsubT}{Y_{\ensuremath{\scriptscriptstyle T}}}

\allowdisplaybreaks

\IEEEoverridecommandlockouts
\title{Realisation-Level Privacy Filtering}
\author{\IEEEauthorblockN{Sophie~Taylor, Praneeth~Kumar~Vippathalla, and Justin~P.~Coon}
\thanks{The authors are with the Department of Engineering Science, University of Oxford, Oxford, U.K (e-mail: sophie.taylor2@balliol.ox.ac.uk; praneeth.vippathalla@eng.ox.ac.uk; justin.coon@eng.ox.ac.uk). This research was funded in whole or in part by the Engineering and Physical Sciences Research Council under grant number EP/W524311/1, and the U. S. Army Research Laboratory and the U. S. Army Research Office under grant number W911NF-22-1-0070. For the purpose of Open Access, the authors have applied a CC BY public copyright license to any Author Accepted Manuscript (AAM) version arising from this submission.}}
\begin{document}
\maketitle
\IEEEpubidadjcol

\begin{abstract}
    We study differentially private data release, where a database is accessed through successive, possibly adaptive queries and mechanisms.
    Existing composition theorems and privacy filters almost always combine worst case per-round privacy parameters, leaving room for more refined accounting based on realised privacy loss,
    which we term realisation-level accounting.
    We present a realisation-level filtering approach to determine stopping times for data releases, and design one such filter.
    Despite technical challenges arising from conditioning on realisations and stopping time, we prove that the filter guarantees $(\epsilon, \delta)$-differential privacy, with $\epsilon$ and $\delta$ chosen by the data handler.
    Through numerical evidence, we demonstrate that
    realisation-level filtering provides a path to better utility beyond mechanism-level methods.
    Furthermore, our proposed filter applies to arbitrary mechanisms, including those that are badly behaved under R{\'e}nyi differential privacy.

    Note: This manuscript supersedes the earlier ISIT conference version. It contains updates and revisions. Readers are encouraged to use and cite this version.
\end{abstract}

\section{Introduction}

In most privacy preserving applications, a database is  subject to successive queries, and is therefore accessed more than once.
In modern settings, queries and mechanisms are often adaptive, meaning they depend on previously released outputs.
A key example is federated learning, where model training involves repeated access to data.
In such systems, the sequence of queries and mechanisms may be generated by an adaptive training procedure.
Moreover, each data access incurs a privacy loss, making adaptive privacy accounting essential.
Hence,
it is crucial that a system designer can quantify how privacy guarantees compose under multiple adaptive mechanism uses.
Throughout this work, we take differential privacy (DP) as the privacy notion.

Existing approaches to privacy composition provide powerful guarantees in many contexts, but can be overly conservative in certain settings.
Classical composition theorems bound the cumulative privacy loss without knowledge of mechanism outputs, commonly by combining known per-mechanism parameters \cite{Dwork2_2006, Dwork3_2010, dwork2014textbook}.
Notably, Rényi differential privacy (RDP) provides  strong composition guarantees by leveraging the distribution of the privacy loss \cite{RDP} and more recent FFT \cite{FFT2020, FFT2021} and saddle-point based \cite{saddle2023} approaches refine this by directly targeting the tail probability of the privacy loss, the quantity that governs DP.\footnote{As typically presented, saddle-point approaches use asymptotic approximations rather than explicit upper bounds, making them more naturally suited to evaluating privacy loss than to guaranteeing strict DP bounds. FFT-based convolution methods rely on discretisation and incur non-constant runtime.}
In practice, RDP is used as an accounting tool, with privacy guarantees converted back to DP for reporting.
Classical guarantees must be computed in advance and are independent of the realised mechanism outputs.
Given a privacy budget, the number of allowable releases must be determined uniformly over all possible mechanism choices, rather than tailored to the specific sequence used.
This can lead to very conservative stopping rules in adaptive scenarios.

To tackle the adaptive setting, researchers have proposed the use of privacy filters, which may be DP based \cite{rogers2021, filters2021} or RDP based \cite{Lecuyer2021Practical, filters2021}, and 
keep a running total of privacy loss to adaptively decide when to stop releases to stay within a privacy budget.
Existing privacy filters track the privacy loss of the sequence of realised mechanisms by combining their per round parameters.
This contrasts classical composition, which operates without knowledge of the particular realisations of adaptively chosen mechanisms.
We call this \emph{mechanism-level} accounting.
Despite its advantages, mechanism-level accounting relies on worst case privacy parameters, and does not exploit the fact that the realised loss may be significantly smaller.
In this work, we present a filtering approach that tracks privacy loss pointwise, operating at the \emph{realisation-level}.
While this may appear natural, doing so poses significant technical challenges, as differential privacy can be violated through conditioning on realised outputs. 
In particular, designing a stopping rule to ensure a privacy guarantee requires accounting for the privacy loss from halting the filter.
Whilst realisation-level privacy accounting was considered in the local differential privacy (LDP) setting in \cite{Pan}, it remains relatively unexplored.

In this work, we propose a privacy filtering approach based on realisation-level accounting. We design one such filter and prove that it satisfies $(\epsilon, \delta)$-DP.
The filter is generally applicable, and does not assume a particular class of mechanism.
Finally, we discuss its utility implications in terms of the allowed number of data releases.

\section{Privacy Filtering}\label{section: setup and notation}

Consider the adaptive data privacy problem setup in Figure \ref{fig: original setup}.
\begin{figure}[!htpb]
  \centering
\includegraphics[scale=0.35]{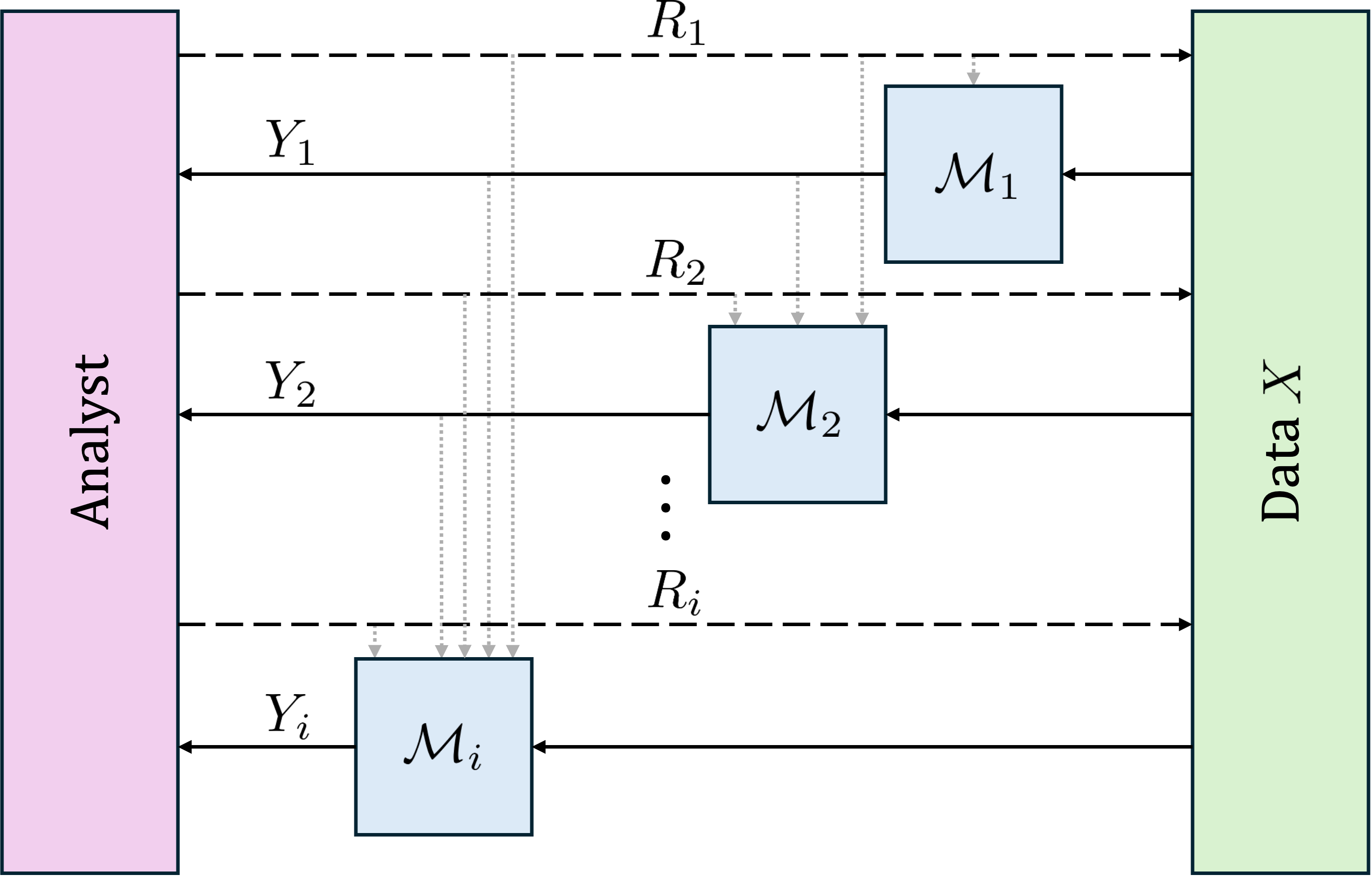}
  \caption{Adaptive data privacy problem}
  \label{fig: original setup}
\end{figure}
An analyst sends a data request $R_1$ from the set of allowable requests $\mathcal{R}_1$ to a database $X$. This is input to a privacy mechanism $\mathcal{M}_1$, which 
produces a random output
$Y_1 = \mathcal{M}_1(R_1, X)$.
Given this response, the analyst can make a second request $R_2 \in \mathcal{R}_2$.
Importantly, requests can be chosen adaptively, depending on previous outputs.
The second data release $Y_2$ is the output of $\mathcal{M}_2(R_1, R_2, Y_1, X)$.
The process continues so that the $i$th data release $Y_i$ is generated by $\mathcal{M}_i( R_1, \dots, R_i, Y_1, \dots, Y_{i-1}, X)$. We assume a fixed series of allowable sets $\mathcal{R}_1, \mathcal{R}_2 \dots$, 
and a fixed family of conditional distributions $P_{\scriptscriptstyle X}(Y_1|R_1), P_{\scriptscriptstyle X}(Y_1|R_1, R_2, Y_2), \dots$ defining the mechanisms $\mathcal{M}_1, \mathcal{M}_2, \dots$ a priori; the mechanisms are adaptive in the standard sense that, at each round, the distribution of the output is conditioned on all previous requests and outputs.

We use superscript indexing to denote the first $i$ requests $R^{i} := (R_1, \dots, R_{i})$ and the first $i-1$ outputs  $Y^{i-1} := (Y_1, \dots, Y_{i-1})$. For $i \geq 1$, we refer to $(R^i, Y^{i-1})$ as the \emph{partial transcript}  and $(R^i, Y^{i})$ as the \emph{full  transcript} at step $i$.
Conditioned on a realised partial transcript $(R^i, Y^{i-1})=(r^i, y^{i-1})$ and database $X=x$, the output $Y_i$ is distributed according to $P_x(y_i|r^i, y^{i-1})$. Throughout, the subscript $x$ indicates that the probabilities are conditioned on the database $X=x$.

To safely carry out this process, the data handler must control the privacy loss incurred by the sequence of released outputs.
Privacy preservation can be achieved through the use of a privacy filter.  A \emph{privacy filter} \cite{rogers2021} is a sequential algorithm that, after every request, decides whether to halt the data release process or proceed with accepting a new request. In other words, it tracks privacy loss to inform a stopping rule.
Let $T$ denote the filter's stopping time which is defined as the index of the last released output when the algorithm halts, and is a random variable. In order for the filter to guarantee $(\epsilon, \delta)$-DP, we need to consider the following definition.

\begin{definition}[$(\epsilon, \delta)$-DP Privacy Filter]\label{def:dp}
    Given a sequence of mechanisms and allowable requests, a privacy filter is said to be $(\epsilon, \delta)$-{DP} if, for every adversary employing a random and adaptive strategy to choose requests, the following holds:
    \begin{align}
    P_x \big( \big(\RT, \YT \big) \in \mathcal{S} \big) \leq e^\epsilon P_{x'}\left( \left( \RT, \YT \right) \in \mathcal{S} \right) + \delta,
    \end{align}
    for all neighbouring databases $x \sim x'$ and all measurable sets $\mathcal{S}$ of full transcripts.
\end{definition}
Note that we use the full transcript in our DP definition, despite the fact that requests are chosen by the adversary. 
This is because all random variables can be jointly used by the adversary to infer about the database.
By requiring privacy to hold for all measurable sets of full transcripts, we ensure that no adversary can distinguish neighbouring databases beyond the $(\epsilon, \delta)$-DP bound, however they select their requests.

\subsection{Mechanism-Level Privacy Accounting} \label{section: ML acounting}
The classical additive composition result of \cite{Dwork2_2006} enables a simple privacy filter. Let $(\epsilon, \delta)$ be the privacy budget. Suppose $(\mathcal{M}_1, \mathcal{M}_2, \dots, \mathcal{M}_{t-1})$ are the adaptive mechanisms applied so far, and at step $t$ a request $r_t$ is made with $\mathcal{M}_{t}$ being the corresponding mechanism. Then, the data handler might employ the so-called \emph{additive privacy filter} \cite{rogers2021} that checks if $\sum_{i=1}^{t}\epsilon_i \leq \epsilon$ and $\sum_{i=1}^{t}\delta_i\leq \delta$, where $(\epsilon_i, \delta_i)$ are the DP parameters of mechanism $\mathcal{M}_i(r^{i}, y^{i-1},X)$.  
If either condition is violated, the data handler will not apply $\mathcal{M}_{t}$ to $X$ and stops accepting new requests. Otherwise, she applies the mechanism, releases the corresponding data, and accepts a new request.

Using an advanced composition result for adaptive mechanisms, \cite[Thm~5.1]{rogers2021} gave an improved privacy filter (in certain regimes) that checks if $\delta' + \sum_{i=1}^{t}\frac{2 \delta_i}{\epsilon_i e^{\epsilon_i}} \leq \delta$ and $\epsilon$ exceeds
\begin{equation} \label{eq: advanced f}
\begin{aligned}
&\sum_{i=1}^t \frac{\epsilon_i \left(e^{\epsilon_i}-1\right)}{2}
+ \sqrt{
2\left(
\sum_{i=1}^t \epsilon_i^2
+ \frac{\epsilon^2}{28.04 \log(\frac{1}{\delta'})}
\right)
}
\\
&\;
\times \sqrt{
\left(
1 + \frac{1}{2}
\log\!\left(
\frac{28.04 \log(\frac{1}{\delta'})\sum_{i=1}^t \epsilon_i^2}{\epsilon^2}
+ 1
\right)
\right)
\log\left(\frac{2}{\delta'}\right).
}
\end{aligned}
\end{equation}
This will be referred to as an \emph{advanced privacy filter}.

Privacy filters based on RDP parameters have also been proposed \cite{Lecuyer2021Practical, filters2021}.
These track privacy loss in terms of R{\'e}nyi divergence of fixed order $\alpha > 1$.
Let $\rho_i(\alpha)$ denote the adaptive RDP parameter of the mechanism $\mathcal{M}_i(r^{i}, y^{i-1},X)$, which may depend on previous requests and outputs.
By known composition properties of RDP \cite{RDP}, the cumulative RDP parameter after $t$ steps is $\sum_{i=1}^t \rho_i(\alpha)$.
To employ an adaptive \emph{RDP privacy filter}, the data handler checks at each step whether $\sum_{i=1}^t \rho_i(\alpha) \leq \epsilon_\alpha$.
If the condition is violated, the filter stops and $\mathcal{M}_t$ is not applied.
From \cite{RDP}, this yields an $(\epsilon, \delta)$-DP guarantee, where $\epsilon = \epsilon_\alpha + \frac{\log\frac{1}{\delta}}{\alpha -1}$.
Importantly, $\alpha$ must in general be fixed before running the filter and cannot be optimised during its execution.
In the mechanism-level setting, the quantities $\rho_i(\alpha)$ are transcript-dependent and do not hold uniformly over all adaptive executions.
In contrast, under classical RDP composition, the composed RDP parameter does not depend on the execution, and an optimisation over $\alpha>1$ can safely be performed.

As these filters use mechanism parameters to track privacy loss, we call them \emph{mechanism-level} privacy filters. All are designed to guarantee $(\epsilon, \delta)$-DP.

\subsection{Realisation-Level Privacy Accounting}
Mechanism-level privacy accounting relies on worst case parameters for each mechanism, regardless of realised outputs. As a result, these filters may halt in executions where the realised privacy loss is significantly less than the budget. Therefore, like \cite{Pan}, we present another paradigm for a privacy filter design based on \emph{realisation-level} privacy accounting. In our work, the idea is to track the accumulated privacy loss of the full transcript $(r^i,y^{i})$ at the $i$th step. The privacy loss incurred in releasing the output $y_i$ at $i$th step given the partial transcript $(r^i,y^{i-1})$ is 
\[
l_i(x,x') := \log\frac{P_x(y_i\mid r^i,y^{i-1})}{P_{x'}(y_i\mid r^i,y^{i-1})}.
\]
The cumulative privacy loss at step $i$ corresponding to the full transcript $(r^i,y^i)$ is given by $l^{(i)}(x,x') \defeq l_1(x,x') + \dots + l_i(x,x')$. The associated random variables are denoted by $L_i(x,x')$ and $L^{(i)}(x,x')$.
For brevity, we will sometimes drop the explicit dependency on $x \sim x'$, so $L_i \defeq L_i(x,x')$.
We refer to these quantities as privacy loss.
A filter tracking the privacy loss this way has the potential to be less conservative by adapting to favourable outcomes, thereby allowing more data releases without compromising privacy. This intuition is formalised by our proposed privacy filter in Sec.~\ref{sec:algo} and the rest of the paper.

Before presenting any realisation-level privacy filters, we consider general aspects of their design. Upon stopping, the full-transcript privacy loss to an adversary is given by the random variable,
\begin{align*}
    \log \frac{P_x( \RT, \YT)}{P_{x'}(\RT, \YT)},
 \end{align*}
where $T$ is the stopping time of the algorithm, and $(\RT, \YT)$ is the full transcript the adversary has. 
By noting that $R_i$ cannot depend on $X$ given $(R^{i-1}, Y^{i-1})$, we can write
 \begin{align*}
    \log &\frac{P_x( \RT, \YT)}{P_{x'}(\RT, \YT)} \nonumber\\
    &=  \log \frac{\prod_{i=1}^{T} P_{x}(R_i|R^{i-1}, Y^{i-1})P_{x}(Y_i|R^{i}, Y^{i-1})}{\prod_{i=1}^{T} P_{x'}(R_i|R^{i-1}, Y^{i-1})P_{x'}(Y_i|R^{i}, Y^{i-1})} \nonumber\\
    &=   \log \frac{\prod_{i=1}^{T} P_{x}(Y_i|R^{i}, Y^{i-1})}{\prod_{i=1}^{T} P_{x'}(Y_i|R^{i}, Y^{i-1})} \nonumber\\
    &= \log\frac{P_x(Y_1|R_1)}{P_{x'}(Y_1|R_1)} + \dots + \log\frac{P_x(\YsubT| \RT, \YTone)}{P_{x'}(\YsubT| \RT, \YTone)}\!\nonumber\\
    &=\LT(x,x').
 \end{align*}
Therefore, the filter bounds the privacy loss to an adversary at the stopping time by tracking and limiting the cumulative privacy loss.

Pan \cite{Pan} proposed a realisation-level filter, termed the \emph{Bayesian Privacy Filter}, in the context of pure LDP. The same construction extends to the $(\epsilon,0)$-DP setting by replacing the notion of neighbouring inputs with the standard DP neighbouring relation. 
Before executing a mechanism, the filter examines every possible output and rejects the query if any would cause the cumulative privacy loss to exceed $\epsilon$. 
The paper also presents an $(\epsilon,\delta)$-LDP construction that can similarly be adapted to the $(\epsilon,\delta)$-DP setting. In this construction, an arbitrary mechanism-level privacy filter is applied first. 
Once this filter stops, the remaining privacy budget is managed using the Bayesian Privacy Filter under a pure privacy guarantee.
In other words, later mechanisms are only executed if the probability of the privacy loss exceeding the remaining budget is zero. 

Motivated by the Bayesian Privacy Filter, one might naturally seek a similar realisation-level stopping rule that directly guarantees general 
$(\epsilon, \delta)$-DP.
In this vein, we consider the following \emph{naive} realisation-level privacy approach.
At each step $i$, upon receiving request $r_i$, it considers the accumulated privacy loss $l^{(i-1)}(x,x')$ from the realised full transcript $(r^{i-1},y^{i-1})$,
and analyses the mechanism $\mathcal{M}_i(r^i,y^{i-1}, X)$.
If the worst-case (over databases) probability that releasing $Y_i$ would cause the accumulated privacy loss to exceed $\epsilon$ is greater than $\delta$, i.e.,
\[ P_{x} \bigl( L_i(x,x') > \epsilon - l^{(i-1)}(x,x') \mid  r^i,  y^{i-1} 
\bigr) > \delta,\]
for any $x \sim x'$,
the filter halts and the mechanism is not executed; otherwise, the mechanism is applied and the output $y_i$ is released.
This construction ensures the stopping time leaks no information about $X$ beyond that jointly revealed by requests and outputs.

Despite 
its intuitive nature and handling of stopping decisions,
$(\epsilon, \delta)$-DP is not guaranteed. 
To see this, consider 
repeated application of an identical non-adaptive binary erasure mechanism,
where an adversary makes the same fixed request at each step, and observes the output of the mechanism.
Here, $Y=X$ with probability $p \leq \delta$ and $Y=\Delta$ otherwise.
These events correspond to infinite and zero privacy loss respectively. 
The filter halts at step $i$ iff $l^{(i-2)}(x,x')=0$ and $l^{(i-1)}(x,x')=\infty$ for one ordering of the neighbouring pair. In other words, the naive filter continues as long as only $\Delta$'s are observed.
Therefore, with probability one, the infinite privacy loss event occurs, and by Def.~\ref{def:dp}, guarantee is no better than $(\epsilon, 1)$-DP.
Owing to the difficulty in ensuring  $(\epsilon, \delta)$-DP guarantee, a realisation-level privacy filter must be carefully designed.

\section{A Realisation-Level Privacy Filter}\label{sec:algo}

We introduce a privacy filter that tracks cumulative privacy loss, and bounds it with a valid stopping rule.
Unlike the construction of Pan \cite{Pan}, whose approximate guarantee is inherited from a preceding mechanism-level privacy filter, our filter incorporates approximate differential privacy directly into the realisation-level stopping rule.
In step $i$, the request $r_i$ is received.
What follows is subtle but essential.
Rather than deciding whether to release $y_i$, the algorithm decides whether it will receive the \emph{following} request $r_{i+1}$,
using knowledge of $r^i$ and $y^{i-1}$.
Specifically, it uses a step-wise parameter $\hat{\delta}_{i+1}$ to assess the risk of releasing $y_{i+1}$ under the worst case request $r_{i+1}$ and database $x$.
Formally, $\hat{\delta}_{i+1}$ is given by
\begin{align}\label{eq: delta hat def}
\inf \Bigl\{ z \!  \in \! [0,1] : \!\!\!
\inf_{\substack{x\sim x' \\r_{i+1} \in \mathcal{R}_{i+1}}} \!\!\!
P_{x} \bigl( Y_i \in \tilde{\mathcal{Y}}_i(z) \mid  r^i,  y^{i-1} 
\bigr) \ge 1 - \theta \Bigr\},
\end{align}
where $\tilde{\mathcal{Y}}_i(z)$ is a function of $x,x'$ and $r_{i+1}$ and is defined as
\begin{align*}
    \left\{ y_i \!: P_{x} \bigl( L_{i+1}(x,x') \!>\! \epsilon - l^{(i)}(x,x') \mid  r^{i+1},  y^i\bigr ) \le z  \right\}.
\end{align*}
Note that the algorithm requires checking the condition $\hat\delta_{i+1} \leq \tilde\delta$. This is most easily done by confirming whether
\begin{align} \label{eq: better check}
    \inf_{\substack{x\sim x' \\r_{i+1} \in \mathcal{R}_{i+1}}} \!\!\!
P_{x} \bigl( Y_i \in \tilde{\mathcal{Y}}_i(\tilde\delta) \mid  r^i,  y^{i-1} 
\bigr) \ge 1 - \theta,
\end{align}
as \eqref{eq: better check} is true if and only if $\hat\delta_{i+1} \leq \tilde\delta$.
Whatever the outcome of this check, the algorithm then executes mechanism $i$ and releases $y_i$.
The privacy filter is outlined in Algorithm \ref{alg: pw with rv stopping rule v2}, where $r_0$ and $y_0$ are taken to be fixed ($r_0 = y_0 = \perp$), and $l_0=0$.
\begin{algorithm}[!htpb]
\label{alg: pw with rv stopping rule v2}
\caption{Realisation-level privacy filter}
\SetAlgoLined
\DontPrintSemicolon
Input: $x, \epsilon$, $\delta$, $\mathcal{M}_1, \mathcal{M}_2, \dots$, $\mathcal{R}_1, \mathcal{R}_2, \dots$

Choose $\tilde{\delta} \in [0, \delta]$, $\theta \in [0,1]$, and $N \in \mathbb{Z}^+$ such that $\tilde{\delta} + \theta (1- \tilde{\delta}) N \leq \delta$

Initialize $i=0$, $l^{(-1)} = 0$

\While{$i \leq N$}{

Receive $r_i$

Execute $\mathcal{M}_i( r^i, y^{i-1}, x)$ to obtain and release $y_i$.

\If{$\hat{\delta}_{i+1} > \tilde{\delta}$ (if \eqref{eq: better check} is false)}{
        \textbf{break}\;
    }

$l^{(i)}(x,x') \leftarrow l^{(i-1)}(x,x') + l_i(x,x'), \forall x \sim x'$

$i \leftarrow i+1$
}

\end{algorithm}

The look-ahead design of the stopping rule makes the filter $(\epsilon, \delta)$-DP, which is proved in Sec.~\ref{section: privacy}.
Formally, it ensures that the stopping event $\{T=i\}$ does not depend on the realised output $y_i$, which allows the stopping time to be decoupled from the privacy loss in the proof of Theorem~\ref{alg: pw with rv stopping rule v2}.

Finally, we remark that the algorithm requires a choice of parameters $(\tilde{\delta}, \theta, N)$. 
Any decision satisfying the condition $\tilde{\delta} + \theta (1-\tilde{\delta}) N \leq \delta$ preserves $(\epsilon, \delta)$-DP, but may yield substantially different stopping times.
Here, $N$ is the maximum stopping time. The parameter
$\theta$ introduces a relaxation in \eqref{eq: better check}, allowing the condition to hold with high probability rather than with probability $1$.
If $\theta = 0$, 
the privacy loss induced by $y_{i+1}$ is analysed for the worst possible realisation of $y_i$. Increasing $\theta$ softens this condition.
In turn $\tilde{\delta}$ must be reduced from $\delta$ to accommodate positive $\theta$.
We also see a trade off between $\theta$ and $N$. The maximum stopping time is allowed to be very large if $\theta$ is very small.
The best parameter choice may vary according to the sequence of mechanisms input to the algorithm.
In Sec. \ref{sec: utility}, we discuss parameter selection for a series of i.i.d Gaussian mechanisms, and find that the $(\tilde\delta, \theta)$ pair can be chosen optimally given $N$.
In fact, the idea may be extended to any set of i.i.d. mechanisms.

\section{Privacy Guarantee} \label{section: privacy}
In this section, we establish the privacy guarantee of our privacy filter. 
For brevity, we use $P_x (\mathcal{S})$ to mean $P_x ( (\RT, \YT ) \in \mathcal{S})$, and $\mathbb{E}_x$ to denote an expectation with respect to $P_x$.

\begin{theorem}
    The privacy filter described in Algorithm \ref{alg: pw with rv stopping rule v2} is $(\epsilon,\delta)$-DP for any sequence of mechanisms.
    \begin{proof}
        Let $\mathcal{A}(x,x') \defeq \{\LT(x,x') \leq \epsilon \}$. To prove the $(\epsilon, \delta)$-DP guarantee, it is enough to show that 
        \begin{align}
        P_x ( \mathcal{A}^c (x,x')) \leq \delta \quad \forall x \sim x', \label{eq: alg proof first part}
        \end{align}
        because of the following argument. 
        If $(\RT, \YT ) \in \mathcal{A}(x,x')$ then
        \(P_x\left(\RT,\YT\right) \leq e^\epsilon P_{x'}\left(\RT,\YT \right)
        \).
        Hence, for any measurable set $\mathcal{S}$, we have 
       \(         
       P_x ( \mathcal{S} \cap \mathcal{A}(x,x') ) \leq e^\epsilon P_{x'} ( \mathcal{S} \cap \mathcal{A}(x,x') )
        \). This combined with \eqref{eq: alg proof first part}
        yields
        \begin{align}
            P_x (\mathcal{S}) &= P_x (  \mathcal{S} \cap \mathcal{A}(x,x')) + P_x (  \mathcal{S} \cap \mathcal{A}^c(x,x')) \nonumber \\
            &\leq e^\epsilon P_{x'} (  \mathcal{S} \cap \mathcal{A}(x,x')) + P_x (  \mathcal{A}^c(x,x') ) \nonumber \\
            & \leq e^\epsilon P_{x'} (  \mathcal{S}) + P_x (  \mathcal{A}^c(x,x') ) \leq e^\epsilon P_{x'} (  \mathcal{S}) + \delta, 
        \end{align}
        for all neighbours $x \sim x'$, proving the $(\epsilon, \delta)$-DP guarantee. 
        
       Now we show \eqref{eq: alg proof first part}. 
       To achieve this, define the event $\mathcal{A}_i(x,x') \defeq \left\{ L^{(i)}(x,x') \leq \epsilon \right\}$ and the random variable   
       $F_i(x,x') \defeq {\mathbf{1}_{\mathcal{A}_i (x,x')}}/{(1 - \hat{\delta}_i)}$ for $i \geq 1$. 
       Note that the stopping event $T=i$ can be written as
       \begin{align*}
           \{T=i\} = \left\{ \hat{\delta}_{1} \leq \tilde{\delta},   \dots, \hat{\delta}_i \leq \tilde{\delta},\ \hat{\delta}_{i+1} > \tilde{\delta} \right\}.
       \end{align*}
        Now consider the expectation
        \begin{align}
            &\mkern -10mu\mathbb{E}_{x} \left[ F_i(x,x') \mathbf{1}_{\{T=i\}} \right] \nonumber \\
             &= \mathbb{E}_{x} \left[ \mathbb{E}_{x}  \left[  \frac{\mathbf{1}_{\mathcal{A}_i(x,x')} }{1-\hat{\delta}_i}\mathbf{1}_{\{T=i\}} \middle| R^i, Y^{i-1} \right] \right] \nonumber \\
             &= \mathbb{E}_{x} \left[ \mathbf{1}_{\{T=i\}} \frac{1}{1- \hat{\delta}_i} P_{x} \left( \mathcal{A}_i(x,x') \mid R^i, Y^{i-1} \right) \right] \label{eq: func} \\
             &= \mathbb{E}_{x} \left[ \mathbb{E}_{x} \left[ \mathbf{1}_{\{T=i\}} Z
             \mid R^i, Y^{i-2} \right] \right] 
             \label{eq: exp expansion},
        \end{align}
        where \eqref{eq: func} uses that fact that $\hat{\delta}_i$ is a function of $R^{i-1}, Y^{i-2}$ and $\mathbf{1}_{\{T=i\}}$ is a function of $R^{i}, Y^{i-1}$, and for brevity, in \eqref{eq: exp expansion} we use $Z \defeq P_x(\mathcal{A}_i | R^i, Y^{i-2}, Y_{i-1}) / (1- \hat{\delta}_i)$. Since $ Z\geq \mathbf{1}_{\{Z \geq 1\}} $,
        \begin{align}
            &\mathbb{E}_{x} \left[ \mathbf{1}_{\{T=i\}} Z | R^i, Y^{i-2} \right] \nonumber \\& \geq
            P_{x} \left(T=i , Z \geq 1 \mid R^i, Y^{i-2} \right)\nonumber
            \\
            & \geq
            P_{x} \left(T=i \mid R^i, Y^{i-2} \right)+P_{x} \left(  \left\{ Z \geq 1 \right\} \mid R^i, Y^{i-2} \right) - 1\nonumber\\
            & \geq P_{x} \left( T=i \mid R^i, Y^{i-2} \right) - \theta , \label{eq: inner exp bounded 1}
        \end{align}
        where \eqref{eq: inner exp bounded 1} follows from the definition of $\hat{\delta}_i$ \eqref{eq: delta hat def}  that given $(R^{i-1}, Y^{i-2})$,  $Z \geq 1$
        with probability at least $1- \theta$ for all $R_i$.
        By combining \eqref{eq: exp expansion}, and \eqref{eq: inner exp bounded 1}, we finally get
        \begin{align}
            \mathbb{E}_{x} \left[ F_i(x,x') \mathbf{1}_{\{T=i\}} \right] &\geq E_{x} \left[ P_{x}\left(T=i \mid R^i, Y^{i-2} \right) - \theta \right] \nonumber \\
            &= P_x(T=i) - \theta \label{eq: B proof 2},
        \end{align}
       for all $x \sim x'$.
       For an upper bound on the expression on the left-hand side of \eqref{eq: B proof 2}, note that if $T=i$, then $1- \hat{\delta}_i \geq 1 - \tilde{\delta}$. Thus,
       \begin{align} \label{eq: B proof 3}
           \mathbb{E}_{x} \left[ F_i(x,x') \mathbf{1}_{\{T=i\}} \right]&=\mathbb{E}_{x} \left[ \frac{\mathbf{1}_{\mathcal{A}_i(x,x') \cap \{T=i\}}}{1- \hat{\delta}_i} \right] \nonumber\\&\leq \frac{1}{1 - \tilde{\delta}} P_x(\mathcal{A}_i(x,x'), T = i).
       \end{align}
       Combining \eqref{eq: B proof 2} and \eqref{eq: B proof 3} yields
       \begin{align} \label{eq: B proof 4}
           P_x(\mathcal{A}_i(x,x') , T=i) \geq (1 - \tilde{\delta})(P_x(T=i) - \theta),
       \end{align}
       for all $x \sim x'$.
       
       We can finally return to $P_x(\mathcal{A}^c(x,x'))$. Following from \eqref{eq: B proof 4} yields
       \begin{align*}
           P_x(\mathcal{A}^c(x,x')) &= 1 - P_x(\mathcal{A}(x,x')) \\
           &= 1 - \sum_{i=0}^N P_x(\mathcal{A}_i(x,x') , T=i)  \\
           &\leq 1 - \sum_{i=1}^{N} (1- \tilde\delta)(P_x(T=i) - \theta) \\
           &= 1 - (1- \tilde\delta) + N \theta(1- \tilde\delta) \\
           &\leq \tilde\delta + \theta (1- \tilde\delta) N \\& \leq \delta,
       \end{align*}
       where we note that $P_x(\mathcal{A}_0(x,x') , T=0) =0$, and the final inequality is true by construction, proving the theorem.
    \end{proof}
\end{theorem}

\section{Utility of the Realisation-Level Filter}
\label{sec: utility}

Beyond ensuring differential privacy, a privacy filter is more useful if it can run for a longer, enabling many data releases without privacy compromise. 
Utility is thus naturally characterised by the stopping time $T$, the number of database queries before access is cut off. We focus on survival probabilities $P_x(T \geq t)$.
In this section, we numerically compare the utility of existing \emph{mechanism-level} filters \cite{rogers2021, filters2021, Lecuyer2021Practical} with our \emph{realisation-level} privacy filter, described by Algorithm \ref{alg: pw with rv stopping rule v2}.

We simulate the additive, advanced and RDP privacy filters outlined in Sec.~\ref{section: ML acounting}, alongside our realisation-level filter, for a sequence of i.i.d. Gaussian mechanisms.
Suppose $x = (x_1, \dots, x_n)$ and $x' = (x_1', \dots, x_n')$ with $x_j, x_j' \in \{0,1\}$ are two neighbouring datasets that differ only in one entry.
We consider counting queries of the form $r(x) = \sum_{j=1}^n x_j$, and i.i.d. Gaussian mechanisms with $Y_i = r(X) + Z_i$, where $Z_i \sim \mathcal{N}(0, \sigma^2)$, and $\sigma=2$.
We do not separately simulate the 
$(\epsilon, \delta)$-DP construction of Pan \cite{Pan}. 
For Gaussian mechanisms, the privacy loss is unbounded, so the probability of exceeding any finite privacy budget is strictly positive. Hence,
the Bayesian Privacy Filter, if invoked, rejects every query. 
Pan's construction therefore reduces to its underlying mechanism-level privacy filter, which is already represented by the mechanism-level baselines considered here.

We note a practical consideration when implementing mechanism-level privacy filters.
For the additive and adaptive DP-based filters, a given mechanism $\mathcal{M}_i(r_i,y_{i-1},X)$ may admit multiple valid $(\epsilon_i,\delta_i)$ pairs, requiring a choice of parameters at each round.
Similarly, a constant R{\'e}nyi divergence order $\alpha >1$ must be chosen for the RDP filter.
In the current setting, mechanisms are i.i.d. and are therefore all known in advance.
This allows us to select parameter values optimally across all rounds.
Accordingly, all mechanism-level simulations use parameter choices that maximise the stopping time.
This process, and hence the computation of mechanism-level stopping times, is outlined in Appendix~\ref{section: appendix MLST}.
Notably, in the i.i.d. setting, mechanism-level filtering reduces to classical composition. In contrast, realisation-level filters retain their adaptive nature.

Recall that we must fix $(N, \tilde\delta, \theta)$ before running the realisation-level filter.
We start by specifying the maximum stopping time $N$, which determines the maximum number of outputs that may be released. 
Unlike the other parameters, $N$ has a direct operational interpretation, and can therefore be chosen based on the intended use of the filter. 
It should be set to a reasonable upper bound on the number of outputs one expects to consider.
For i.i.d. Gaussian mechanisms, we can then choose $(\tilde\delta, \theta)$ optimally to maximise the stopping time.
\begin{proposition} \label{prop: equiv condition}
    In the i.i.d. Gaussian setting described above with unit sensitivity, i.e., $|r(x)-r(x')|=1 \; \forall x \sim x'$, the condition \eqref{eq: better check} for releasing $y_{i+1}$ is equivalent to $l^{(i-1)}(x,x') \leq 
    \kappa,\; \forall x \sim x'$, where
    \begin{align*}
        \kappa \defeq \epsilon - \frac{1}{\sigma^2} - \frac{1}{\sigma}\left( \Phi^{-1}(1- \tilde\delta) + \Phi^{-1}(1-\theta) \right),
    \end{align*} 
    and $\Phi$ is the cumulative density function of the standard normal distribution.
\end{proposition}
The proof of Proposition \ref{prop: equiv condition} can be found in Appendix~\ref{section: proof of equiv prop}.
Recall that there are only two neighbouring datasets within our setting, and that $l_i(x,x') = - l_i(x',x)$.
Therefore, from Proposition \ref{prop: equiv condition}, we can say
\begin{align*}
    P_x(T \geq t) \!=\! P_x \! \left( |L^{(1)}(x,x')| \leq \kappa, \dots, |L^{(t-2)}(x,x')| \leq \kappa \right) \!.
\end{align*}
Clearly, maximising $\kappa$ maximises the stopping time.
Taking $\theta = \frac{\delta - \tilde\delta}{N(1- \tilde\delta)}$, the optimal $(\tilde{\delta}, \theta)$ pair can be found by numerically minimising $\Phi^{-1}(1- \tilde\delta) + \Phi^{-1}(1-\theta)$.

Following this procedure, we set $(N, \tilde\delta, \theta)$ as $(48, 5.60 \times 10^{-4}, 9.17 \times 10^{-6})$, which yields the result in Fig. \ref{fig: simulations}.
\begin{figure*}[t]
  \centering
\includegraphics[scale=0.6]{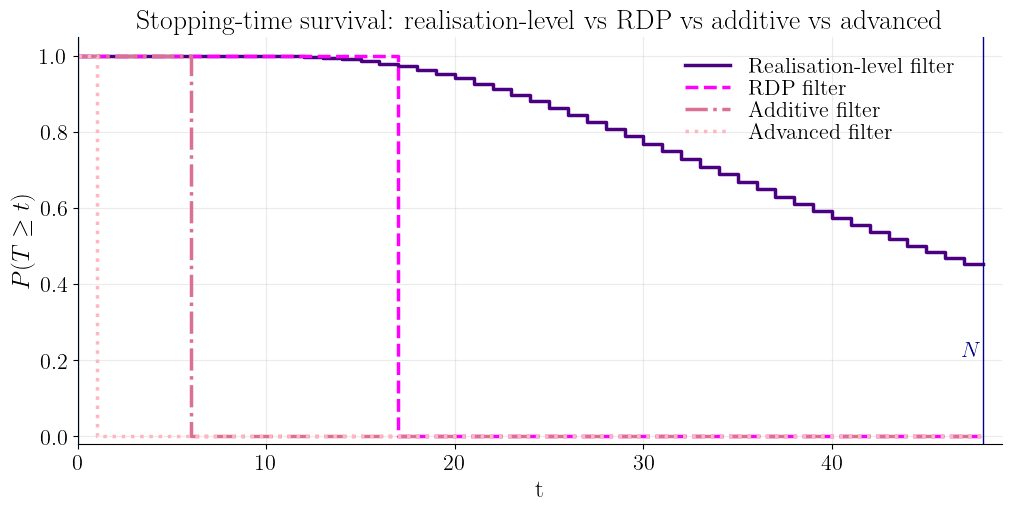}
  \caption{Stopping time survival $P(T \geq t)$ of mechanism-level privacy filters compared with our realisation-level privacy filter.}
  \label{fig: simulations}
\end{figure*}
We remark that, due to symmetry in our setting, $P_x(T \geq t) = P_{x'}(T \geq t)$.
Therefore, we simply plot $P(T \geq t)$ against $t$.
The primary observation is that the realisation-level filter generally dominates the mechanism-level filters, particularly for larger $t$.
The realisation-level filter admits a marginally lower survival probability than the RDP filter for a small number of early $t$ values.
This is consistent with its reliance on realised privacy loss values, which may be unusually high, occasionally triggering early stopping.
The effect is strongly outweighed in the long term.
We also remark the poor performance of the advanced filter, which is down to the particular mechanisms simulated.
The advanced composition theorem significantly outperforms additive composition in the small $\epsilon_i$ regime, which is not where this simulation operates.
Overall, the results highlight the potential for utility gains from privacy accounting at the realisation level.

\section{Discussion and Future Direction}

In this work, we developed privacy filtering at the realisation level, designed one such privacy filter, and proved that it guarantees 
$(\epsilon,\delta)$-differential privacy. 
Unlike mechanism-level approaches, our filter tracks realised privacy loss rather than composing worst-case mechanism parameters.
Previous realisation-level work \cite{Pan} considered pure privacy, and extended to approximate privacy by combining a pure LDP or DP filter with a mechanism-level privacy filter. 
In our work, the stopping rule itself is derived directly for general $(\epsilon, \delta)$-differential privacy.
It can be applied to arbitrary mechanisms, whether they be independent or adaptive, with discrete or continuous outputs that may have differing supports.
This extends beyond settings where R{\'e}nyi divergence is well behaved, as required by RDP filters.
This work demonstrates that realisation-level accounting provides a relatively unexplored axis for utility improvement, by exploiting knowledge of the realised privacy loss.

The main implementation challenge of the filter is in checking the condition $\hat{\delta}_{i+1} > \tilde\delta$.
Whilst analytical computation is perfectly feasible in simple settings, efficient approximations or bounds may be required for more elaborate mechanisms.
Recall the definition \eqref{eq: delta hat def} of $\hat\delta_{i+1}$  and the equivalent check \eqref{eq: better check}.
Given a triple $(x,x', r_{i+1})$, let $M$ denote the number of basic operations required to determine whether $P_{x} \bigl( Y_i \in \tilde{\mathcal{Y}}_i(\tilde\delta) \mid  r^i,  y^{i-1} 
\bigr) \ge 1 - \theta$.
Then, the number of basic operations required to check $\hat\delta_{i+1} > \tilde\delta$ is at most $C=\sum_{x \in \mathcal{X}}\mathcal{N}(x) | \mathcal{R}_{i+1} | M$, where $\mathcal{N}(x)$ counts the number of neighbours of database $x$.
In the i.i.d.\ Gaussian example of Sec. \ref{sec: utility}, $M=1$ and $C = 2$. 
In general however, this bound can be much larger. For complex, non-Gaussian mechanisms, the probability calculations may require Monte Carlo simulation with large sample sizes, which can substantially increase $M$.
For the RDP filter, the main source of computational complexity lies in evaluating $\rho_i(\alpha)$ for each mechanism, where $\rho_i(\alpha)$ must bound the R{\'e}nyi divergence of order $\alpha$
for all pairs of neighbouring datasets. 
This is fundamentally hard when the divergence itself admits no closed-form expression.
Existing filters face a similar scaling issue to ours, as their bounds must hold over all neighbouring datasets. This can become costly as the domain grows, especially for complex mechanisms.
Our filter additionally requires bounds to hold for all requests in $\mathcal{R}_{i+1}$. 
Unlike datasets, this set may be controlled by the data handler, who can reject other requests and safely continue.
The effect of this strategy on utility warrants further study, particularly in adaptive learning settings where requests may be difficult to predict.

Another challenge in implementing our realisation-level filter is parameter selection.
The filter requires prior specification of $(\tilde{\delta}, \theta, N)$.
Proposition \ref{prop: equiv condition} reveals that, in the i.i.d. Gaussian setting, the stopping condition is equivalent to the accumulated privacy loss exceeding a threshold $\kappa$.
In fact, this extends to general i.i.d. mechanisms for suitable $\kappa$.
Thus, for fixed $N$, maximising $\kappa$ over $(\tilde\delta, \theta)$ yields the maximum stopping time for any transcript.
General optimal selection of filter parameters remains open.
One possible approach is to approximate the sequence of mechanisms by an i.i.d. model, and select parameters accordingly.
This may offer useful intuition and a principled alternative to arbitrary parameter selection. However, for mechanisms exhibiting strong dependence, an i.i.d. model may be inaccurate.

Exploring filters that combine realisation-level accounting with Rényi-based composition is a promising direction for future work.
R{\'e}nyi differential privacy is known to provide very strong composition guarantees for light-tailed mechanisms, such as those based on Gaussian noise, and existing RDP filters leverage this property to achieve high utility for some common mechanisms.
The simulation results in this work illustrate that realisation-level accounting can yield significantly improved stopping time behaviour, even over RDP-based methods, by exploiting knowledge of realised privacy loss.
In summary, our work provides evidence that realisation-level accounting provides a complementary path to higher utility.

\appendices
\section{Proof of Proposition \ref{prop: equiv condition}} \label{section: proof of equiv prop}
For i.i.d. Gaussian mechanisms with unit sensitivity, the privacy loss variables are themselves i.i.d. and Gaussian. 
In particular, since $Y_i = r(X) + Z_i$ with $Z_i \sim \mathcal{N}(0, \sigma^2)$, it follows that for any neighbouring inputs $(x,x')$,
\[
L_i(x,x') \sim \mathcal{N} \left( \frac{1}{2 \sigma^2}, \frac{1}{\sigma^2} \right),
\]
under the distribution $P_x$ of $Y_i$.
Recall from \eqref{eq: better check} that at step $i$, Algorithm \ref{alg: pw with rv stopping rule v2} decides to release $y_{i+1}$ if
\begin{align*}
    \inf_{\substack{x\sim x' \\r_{i+1} \in \mathcal{R}_{i+1}}} \!\!\!
P_{x} \bigl( Y_i \in \tilde{\mathcal{Y}}_i(\tilde\delta) \mid  r^i,  y^{i-1} 
\bigr) \ge 1 - \theta,
\end{align*}
where $\mathcal{Y}_i(\tilde\delta)$ is defined as
\begin{align} \label{eq: yi tilde def appendix}
    \left\{ y_i \!: P_{x} \bigl( L_{i+1}(x,x') \!>\! \epsilon - l^{(i)}(x,x') \mid  r^{i+1},  y^i\bigr ) \le \tilde\delta  \right\}.
\end{align}
Let us first examine $\tilde{\mathcal{Y}}_{i}(\tilde\delta)$, dropping the explicit dependencies on $x \sim x'$ for brevity.
The inner probability in \eqref{eq: yi tilde def appendix} is
\begin{align} \label{eq: y tilde inner prob}
    P_{x} \bigl( L_{i+1} > \epsilon - l^{(i)} \mid  l^{(i-1)} \bigr ) = 1 - \Phi \left( \sigma\left( \epsilon - l^{(i)} \right) - \frac{2}{\sigma} \right),
\end{align}
where we have used the fact that mechanisms are i.i.d. to drop the conditional dependence on $r^{i+1}$ and $y^i$. 
Using \eqref{eq: y tilde inner prob}, we find that the condition $y_i \in \tilde{\mathcal{Y}}_i(\tilde\delta)$ 
is equivalent to 
\begin{align*}
    l^{(i)} \leq \epsilon - \frac{1}{2 \sigma^2} - \frac{1}{\sigma} \Phi^{-1}(1- \tilde\delta).
\end{align*}
Substituting this back into the stopping rule \eqref{eq: better check} yields the following condition for the release of $y_{i+1}$:
\begin{align}
\inf_x
P_{x} \left(  L^{(i)} \leq \epsilon - \frac{1}{2 \sigma^2} - \frac{1}{\sigma} \Phi^{-1}(1-\tilde\delta) \mid l^{(i-1)}
\right) \ge 1 - \theta .
\end{align}
Expanding the Gaussian probability on the left hand side and using $L^{(i)} = l^{(i-1)} + L_i$ yields
 \begin{align*}
     l^{(i-1)} \leq \epsilon - \frac{1}{\sigma^2} - \frac{1}{\sigma}\left( \Phi^{-1}(1-\tilde\delta) + \Phi^{-1}(1-\theta) \right),
 \end{align*}
$\forall x \sim x'$, which gives Proposition \ref{prop: equiv condition}.

\section{Stopping Time Computation for Mechanism-Level Filters}
\label{section: appendix MLST}
Let $T_{a}$, $T_{av}$ and $T_{RDP}$ represent the stopping times for the additive, advanced, and RDP filters respectively, as outlined in Sec. \ref{section: ML acounting}.
Recall that each mechanism $\mathcal{M}_i$ is i.i.d.
Throughout this section, we use $\epsilon_i$ as shorthand for $\epsilon_i(\delta_i)$, which is the smallest $\epsilon_i$ for which $\mathcal{M}_i$ is $(\epsilon_i, \delta_i)$-DP.
Then, the stopping time for the additive filter is
\begin{align*}
    T_{a} = \max \left\{ t \in \mathbb{N}_0: \min_{\delta_i \leq \frac{\delta}{t}} \epsilon_i \leq \frac{\epsilon}{t} \right\}.
\end{align*}
Similarly, the stopping time for the advanced filter is
\begin{align*}
    T_{av} = \max \left\{ t \in \mathbb{N}_0: \min_{\delta_i, \delta': \delta' + \frac{2 t \delta_i}{\epsilon_i e ^{\epsilon_i}} \leq \delta} f(\delta_i, \epsilon_i, \delta') \leq \epsilon \right\},
\end{align*}
where we use $f(\delta_i, \epsilon_i, \delta')$ to denote the expression in \eqref{eq: advanced f}.
Finally, for the RDP filter applied to Gaussian mechanisms, the parameter $\rho_i(\alpha)$ admits a closed form expression \cite{RDP}. 
In particular, our example with unit sensitivity yields $\rho_i(\alpha) = \frac{\alpha}{2 \sigma^2}$.
Therefore, the stopping time for the RDP filter is
\begin{align*}
    T_{RDP} = \max \left\{ t \in \mathbb{N}_0: \min_{\alpha > 1} \left( \frac{t \alpha}{2 \sigma^2} + \frac{\log \frac{1}{\delta}}{\alpha - 1} \right) \leq \epsilon \right\}.
\end{align*}
In the implementation of the simulation in Sec. \ref{sec: utility}, $T_{a}$, $T_{av}$ and $T_{RDP}$ were computed numerically.

\bibliographystyle{IEEEtran}
\bibliography{PathDP_ISIT}
\end{document}